\documentclass{article}
\usepackage{ijcai20}
\usepackage{times}
\usepackage{framed}
\usepackage{multirow}
\usepackage{bbding}
\usepackage{diagbox}
\usepackage{booktabs}
\usepackage{threeparttable}
\usepackage[ruled,vlined,linesnumbered]{algorithm2e}
\usepackage{amsfonts}
\usepackage{amsthm}
\usepackage{mathrsfs}
\usepackage{bm}

\usepackage{soul}
\usepackage{url}
\usepackage[hidelinks]{hyperref}
\usepackage[utf8]{inputenc}
\usepackage[small]{caption}
\usepackage{graphicx}
\usepackage{amsmath}
\usepackage{booktabs}
\title{Incentive Mechanism Design for ROI-constrained Auto-bidding}

\author{
	Bin Li\textsuperscript{\rm 1}, 
	Xiao Yang \textsuperscript{\rm 2},
	Daren Sun \textsuperscript{\rm 3},
	Zhi Ji \textsuperscript{\rm 3},
	Zhen Jiang \textsuperscript{\rm 3},
	Cong Han \textsuperscript{\rm 3},
	Dong Hao \textsuperscript{\rm 1} \\
\affiliations
	\textsuperscript{\rm 1} University of Electronic Science and Technology of China \\
\textsuperscript{\rm 2} Independent Researcher \\
\textsuperscript{\rm 3} Baidu Inc. \\
\emails
libin@std.uestc.edu.cn, shawcsn@gmail.com, \{sundaren,  jizhi, jiangzhen01, hancong01\}@baidu.com, haodong@uestc.edu.cn
}

\begin{document}

\maketitle

\begin{abstract}	
	Auto-bidding plays an important role in online advertising and has become a crucial tool for advertisers and advertising platforms to meet their performance objectives and optimize the efficiency of ad delivery. Advertisers employing auto-bidding only need to express high-level goals and constraints, and leave the bid optimization problem to the advertising platforms. As auto-bidding has obviously changed the bidding language and the way advertisers participate in the ad auction, fundamental investigation into mechanism design for auto-bidding environment should be made to study the interaction of auto-bidding with advertisers. In this paper, we formulate the general problem of incentive mechanism design for ROI-constrained auto-bidding, and carry out analysis of strategy-proof requirements for the revenue-maximizing and profit-maximizing advertisers. In addition, we provide a mechanism framework and a practical solution to guarantee the incentive property for different types of advertisers.
\end{abstract}
\newtheorem{lemma}{Lemma}
\newtheorem{defn}{Definition}
\newtheorem{prop}{Proposition}
\newtheorem{corollary}{Corollary}
\newtheorem{example}{Example}
\newtheorem{theorem}{Theorem}

\section{Introduction}
Online advertising has been an indispensable part of modern advertising market. According to the report of IAB, full-year online advertising revenue reached \$124.6 billion in 2019. An important reason for wide-spread adoption of online advertisement comes from its high return on investment (ROI) for advertisers, compared to other traditional marketing methods \cite{moran2005search}. Every day, tens of billions of auctions are held in real time to decide which advertisers' ads are shown, how these ads are arranged, and what the advertisers are charged. 
In the traditional setting, advertisers manually bid on the relevant ads to target their audiences. To keep and improve campaign performance, advertisers should adjust the bid prices frequently according to various bidding feedback signals such as cost-per-click, ad impressions or clicks, and so on \cite{yang2019aiads}. 
However, it is not a trivial matter to carry out the bidding optimization, especially for small and local businesses that do not have dedicated marketing staff. 

Due to the limitation of manual optimization, auto-bidding is becoming a prevalent and advanced tool for advertisers to maximize results based on their campaign goals. Auto-bidding strategy can automatically set bids for ads based on that ad's likelihood to result in a click or conversion. The motivation behind auto-bidding is performance and product simplicity \cite{aggarwal2019autobidding}. Instead of asking advertisers for fine-grained bids, the advertising platform only asks them to submit higher level objectives and constraints. Then the auto-bidding algorithms set unique bids for each auction, based on information available at the time of the auction. Leveraging the machine learning models and real-time control strategies, auto-bidding can provide a more efficient ad delivery and a more steady guarantee for advertisers' constraints.

\begin{figure}[t]
	\centering
	\includegraphics[width=0.8\columnwidth]{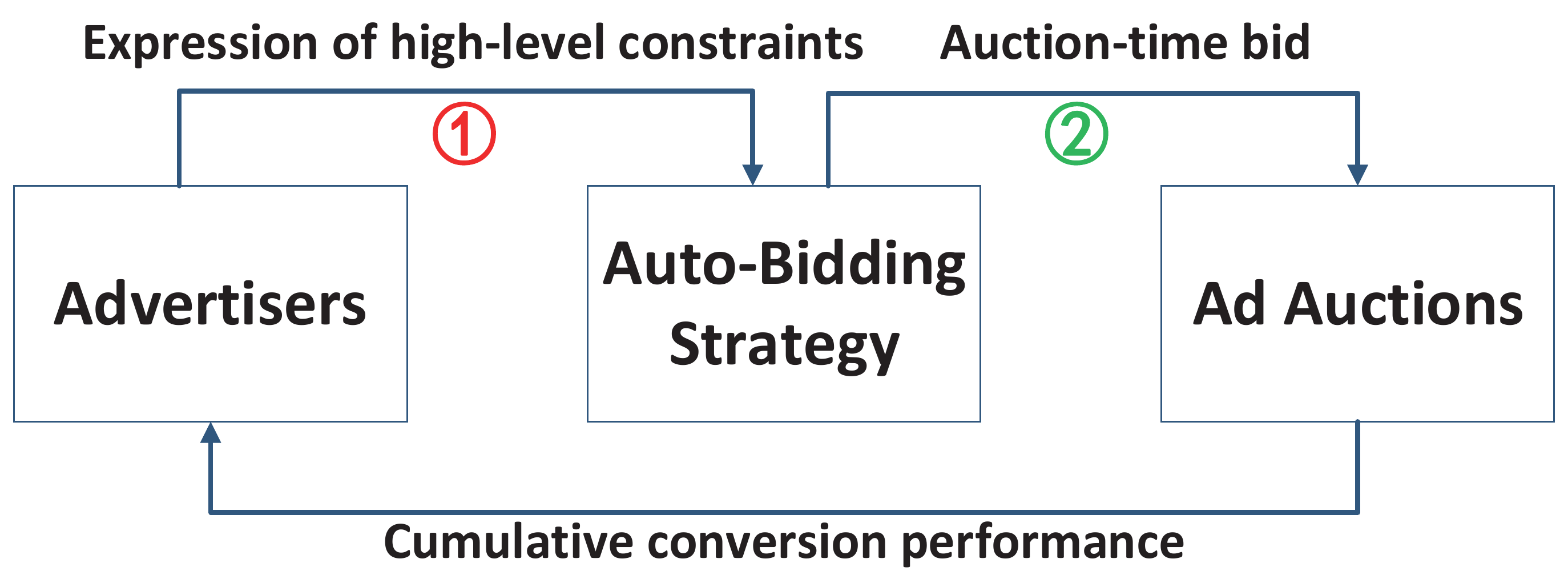}
	\caption{The procedure of auto-bidding.}
	\label{fig1}
\end{figure}

The general procedure of auto-bidding is presented in Figure \ref{fig1}. It is a two-stage bidding expression. At the first stage, advertisers should set their constraints or objectives with high level expressivity such as a target cost-per-acquisition (CPA) or return on investment (ROI). After that, at the second stage, based on the advertisers' reports, the auto-bidding strategy converts these goals and constraints into per-auction bids at serving time to participate in the ad auction, based on predictions of performance of each potential ad impression.  
Apparently, the auto-bidding procedure completely changed the bidding language and the interaction between the advertisers and the advertising platforms. 
Moreover, the advertisers are less sensitive to their obtained performance in a single auction but raise more concerns about their cumulative performance across many auctions.
As the traditional mechanism design for online advertising mostly focuses on the direct interaction between bidders and sellers, the auto-bidding paradigm poses new challenges in designing delivery strategies. If there is no strong incentive guarantee in the basic auto-bidding environment, advertisers may not truly express their private objectives and constraints. Based on false information, the advertising platform is not able to fully optimize advertisers' objectives and it is also hard for the advertisers to reason about the optimal report strategy.

In this paper, we formulate the general problem of incentive mechanism design for ROI-constrained auto-bidding, and analyze the strategy-proof requirements for the revenue-maximizing and profit-maximizing advertisers. Specifically, we obtain the solutions for two common expressions of the ROI constraints. For advertisers with a target CPA expression, we provide a full characterization of incentive-compatible (IC) delivery mechanisms and build a double-layer framework to implement the IC mechanisms in reality. For advertisers with a target ROI expression, two classes of delivery mechanisms are proposed so as to satisfy advertisers' specific needs. The results developed in this paper contribute to the literature of auto-bidding and provide some insights into auto-bidding product design.

\subsection{Related Literature}
{\bf Bidders with constraints.} The literature on bidders with constraints falls into two classes. The first class focuses on the optimal auction design. \cite{laffont1996optimal} is the first work in this line, which studies symmetric bidders with budget constraint. The authors show that the optimal auction is the one with a reserve price. After this work, a series of paper such as \cite{che2000the,malakhov2008optimal,pai2014optimal} improve their result and extend it into other directions. 
\cite{golrezaei2020roi} investigates optimal auction design for bidders with target ROI constraints.
The second class considers the optimal bidding strategy of constrained bidders \cite{goel2009hybrid,heymann2019cost,golrezaei2020roi,aggarwal2019autobidding}. In particular, \cite{heymann2019cost} studies bidders with target CPA constraints and derives a Nash equilibrium for the second price auction. \cite{golrezaei2020roi} presents that bidders with ROI constraints tend to shade their bids in auctions. 

{\bf Bidding optimization.} Bidding optimization has been well studied in online advertising, especially in sponsored search \cite{dar2009bid,broder2011bid,zhao2018deep}. A systematic study of a family of greedy strategies is undertaken in \cite{cary2007greedy}. For advertisers with ROI constraints, \cite{rey2010conversion} proposes a conversion-based bid strategy to automatically adjust the bid price. For display advertisers, a bid strategy called optimized cost per click is proposed in \cite{zhu2017optimized}. Very recently, \cite{aggarwal2019autobidding} proposes a general framework to construct optimal bidding strategies for advertisers with multiple affine constraints in multi-slot truthful auctions. There is also a lot of work concentrating on budget optimization \cite{feldman2007budget,zhang2012joint,agarwal2014budget,lee2013real,karande2013optimizing}, where the goal is to maximize bidders' utilities by spending their budgets smoothly over the day. To satisfy different objectives from advertisers, online advertising platforms \cite{google:auto,microsoft:auto,facebook:auto} have provided various types of auto-bidding products such as oCPM, oCPC (Target CPA), Target ROAS (Target ROI), Maximize Conversions, etc. To the best of our knowledge, we present the first work that considers the incentive problem for the auto-bidding.

The reminder of this paper is organized as follows. We first introduce the basic setting of auto-bidding environment and formulate the general incentive problem for ROI-constrained advertisers. Then we study delivery mechanisms that can incentivize advertisers to reveal their true target CPA constraints, and propose practical solutions that are able to be carried out in reality. After that, we turn to advertisers with target ROI expressions, and propose two families of incentive-compatible delivery mechanisms.

\section{Preliminaries}
\subsection{Basic Setting}
Assume that there is a set of advertisers $A$. For a period of time $T$, there is a set of ad request $R$, each with potential multiple impressions. Given an advertiser $a\in A$ employing auto-bidding, denote $R_a \subseteq R$ by the ad request set that is related to $a$. Associated with each ad request $j\in R_a$ are a set of available slots $S_j$ and a set of matched advertisers $N_j$ ($a\in N_j$). Besides, for $i\in N_j$, denote $ctr_i^{js}$ by the {\it click-through rate} (ctr) of $i$'s ad in slot $s\in S_j$ and let the {\it conversion rate} (cvr), the probability that a conversion, like an app installation, is acquired after click this ad be $cvr_i^j$. If a conversion does occur in ad request $j$, then $i$ receives a revenue $v_i^j$ which is from a finite set $V_i$.
When an ad request $j$ arrives, there is an auction to determine which advertisers' ads are displayed and what prices each advertiser is charged.

Let $b_j=(b_i^j)_{i\in N_j}$ be the bid profile of all advertisers in $j$ and $b_{-i}^j$ be the bid profile except $i$, i.e., $b_j=(b_i^j,b_{-i}^j)$. The ad auction computes the allocations $\{x_i^{js}(b_j)\}_{i\in N_j, s\in S_j}$ and the payments $\{p_i^{js}(b_j)\}_{i\in N_j,s\in S_j}$, where $x_i^{js}(b_j)$ is the indicator variable if slot $s$ of impression $j$ was allocated to $i$ and $p_i^{js}(b_j)$ is the price paid by $i$ if her ad is clicked. That is, if advertiser $i$ wins some slot $s$ in ad request $j$, then $x_i^{js}(b_j)=1$ and if her ad is further clicked by the user, she pays $p_i^{js}(b_j)$ to the advertising platform (the CPC model).

Since advertiser $a$ adopts auto-bidding, a proxy advertiser $\tilde{a}$, controlled by the advertising platform, will bid in each auction based on $a$'s objectives and predictions of performance of each impression. We use $cv_a^{j}(b_j)=\sum_{s}x_a^{js}(b_j)\cdot ctr_a^{js} \cdot cvr_a^j$ to denote the expected conversion of advertiser $a$ in ad request $j$ and $cpa_a^{j}(b_j)=\sum_{s}p_a^{js}(b_j)/{cv_a^{j}(b_j)}$ to denote the cost-per-acquisition (CPA) which represents the average cost for each conversion. Let $(cpa_a^j,cv_a^j)$ be the delivery function and $(cpa_a^j(b_j),cv_a^j(b_j))$ be an actual delivery in ad request $j$. 

Commonly, there are two goals for performance advertisers. The first one is optimizing the total revenue $U_a^r$, i.e., the cumulative conversion values obtained in the auto-bidding episode:
\begin{alignat}{1}
U_a^r &=\max_{\{b_a^j\}} \sum_{j\in R_a}v_a^j\cdot cv_a^j(b_j).
\end{alignat}
The second goal is maximizing the profit $U_a^p$:
\begin{alignat}{1}
U_a^p &=\max_{\{b_a^j\}} \sum_{j\in R_a}(v_a^j-cpa_a^j(b_j))\cdot cv_a^j(b_j).
\end{alignat}
These two goals can be formulated into a single utility function $U_a=\max_{\{b_a^j\}} \sum_{j\in R_a}(v_a^j-I_a\cdot cpa_a^j(b_j))\cdot cv_a^j(b_j)$, where $U_a=U_a^r$ if $I_a=0$ and $U_a=U_a^p$ for $I_a=1$.

As performance advertisers optimize the immediate tradeoff between conversions and costs generated directly from their ads, return on investment (ROI) has been the standard metric for measuring this tradeoff across all types of advertising \cite{borgs2007dynamics,auerbach2008an,wilkens2016mechanism}. ROI measures the ratio of the profit obtained (``return'') to the cost or price paid (``investment''), i.e., the density of profit in cost: $ROI = \frac{revenue {-} cost} {cost}$. In general, performance advertisers want to maximize their utilities with an ROI constraint. For the ROI constraint $ROI_a\ge \gamma_a$, we have:
\begin{alignat}{2}
& \frac{revenue {-} cost} {cost} \geq \gamma_a \\
& \Leftrightarrow \frac{\sum_{j\in R_a}(v_a^j-cpa_a^j(b_j))\cdot cv_a^j(b_j)}{\sum_{j\in R_a}cpa_a^j(b_j)\cdot cv_a^j(b_j)}\ge \gamma_a\\
& \Leftrightarrow cpa_a(\{b_j\})\le \frac{1}{1+\gamma_a}\cdot \frac{\sum_{j\in R_a}v_a^j\cdot cv_a^j(b_j)}{\sum_{j\in R_a}cv_a^j(b_j)},
\end{alignat}
where $cpa_a(\{b_j\})=\sum_{j}cpa_a^j(b_j)\cdot cv_a^j(b_j)/\sum_{j}cv_a^j(b_j)$ is the CPA delivered in the auto-bidding episode.

If advertiser $a$'s value-per-conversion $v_a^j$ is the same value for all ad requests, i.e., $v_a^{j_1}=v_a^{j_2}=v_a$ for all $j_1,j_2 (j_1\ne j_2)\in R_a$, then a target ROI (tROI) constraint can be transformed to a target CPA (tCPA) constraint:
\begin{alignat}{2}
roi_a\ge \gamma_a \Leftrightarrow cpa_a(\{b_j\})\le \frac{v_a}{1+\gamma_a}=t_a.
\end{alignat}
For this case, an advertiser with a tROI $\gamma_a$ is equivalent to an advertiser with a tCPA $t_a=\frac{v_a}{1+\gamma_a}$. In this paper, $\gamma_a$ is assumed as a bottom-line. That is, if the actual ROI in the auto-bidding episode does not reach the minimum requirement $\gamma_a$, then $U_a< 0$.

In practice, tCPA and tROI are two common ways to express advertisers' requirements. For conversions like app installations or sign-ups, a tCPA expression is more eligible to evaluate the delivery effects, as the changes of conversion values in different ad requests are relatively small. In other cases where the conversion value varies with ad requests, such as shopping sales, the advertising platform needs to learn tROI constraints to better optimize advertisers' goals. We denote $c_a$ by $a$'s constraint expression (a tROI $\gamma_a$ or a tCPA $t_a$) and specific its meaning in the context. To implement an auto-bidding product, the advertising platform needs to resolve two problems. The first one is to design and optimize the bidding algorithms to meet advertisers' objectives and constraints, the second one is to design delivery mechanisms to ensure the incentive property. 

\subsection{Optimal Bidding Strategy}
From the view of advertiser $a$, the whole auto-bidding episode is totally a black box, which takes her expression $c_a$ as input and outputs the ultimate ad delivery, such as the total conversions and spending. During the auto-bidding episode, the proxy advertiser $\tilde{a}$ tries to solve following optimization problem:

\begin{alignat}{1}
&\max_{\{b_a^j\}} \sum_{j\in R_a}(v_a^j-I_a\cdot cpa_a^j(b_j))\cdot cv_a^j(b_j),\\
&s.t.\quad ROI_a\ge \gamma_a \quad or \quad CPA_a\le t_a.
\end{alignat}

In the paper of \cite{aggarwal2019autobidding}, the authors build a general framework to compute $\tilde{a}$'s optimal bidding strategies though LP duality. In particular, given advertiser $a$'s goal and predictions of other advertisers' bid profiles, they show that there is an optimal bidding formula which takes in $a$'s value-per-conversion, and converts it into a bid in each ad request. For our scenario, we collect the results in Table 1. For example, if the value-per-conversion is unvaried for different conversions, then the optimal bid for a revenue-maximizing advertiser with a tCPA $t_a$ is $B_a\cdot cvr_a^j$, where $B_a=t_a+\frac{v_a}{m}$ is a constant.

\begin{table}[t]
	\centering
	\begin{tabular}{|m{1.4cm}<{\centering}|m{3.1cm}<{\centering}|m{2.58cm}<{\centering}|}
		\hline
		& Revenue-maximizing & Profit-maximizing \\ \hline
		tCPA & $\frac{(m \cdot t_a+v_a^j)\cdot cvr_a^j}{m}$ & $\frac{(m\cdot t_a+v_a^j)\cdot cvr_a^j}{m+1}$       \\ \hline
		tROI      &$\frac{(n+\gamma_a+1)\cdot v_a^j\cdot cvr_a^j}{n\cdot (\gamma_a+1)}$   & $\frac{(n+\gamma_a+1)\cdot v_a^j\cdot cvr_a^j}{(n+1)\cdot (\gamma_a+1)}$      \\ \hline
	\end{tabular}
	\caption{Optimal bidding strategies for different goals. $m$ and $n$ are two predetermined parameters.}
\end{table}
\subsection{Incentive Problem}
Given predictions of other advertisers' bid profiles, since the proxy advertiser $\tilde{a}$ is always committing an optimal bidding strategy, it is straightforward that the delivery function $\{(cpa_a^j,cv_a^j)\}_{j\in R_a}$ is only parameterized by $c_a$. For any expression $c_a$, $cpa_a^j(c_a)$ and $cv_a^j(c_a)$ are the expected delivery in ad request $j$. Now we formally define the delivery mechanism of the advertising platform and the related concepts in the rest of this section.
\begin{defn}
	A delivery mechanism is defined by a set of delivery functions $\{(cpa_a^j,cv_a^j)\}$, where $cpa_a^j: R\rightarrow R$ and $cv_a^j:R\rightarrow [0,1]$ are the CPA function and the conversion function for ad request $j$, respectively.
\end{defn}

Given a delivery mechanism $\{(cpa_a^j,cv_a^j)\}$, $a$'s utility $U_a$ can be expressed as
\begin{alignat}{1}
\sum_{j\in R_a}(v_a^j-I_a\cdot cpa_a^j(c_a))\cdot cv_a^j(c_a).
\end{alignat}
As $c_a$ is private information, if the delivery mechanism is not properly designed, advertiser $a$ has incentives to misrepresent her constraint to increase her own utilities. There are two problems when such behaviors arise. Firstly, the advertising platform is not able to fully optimize $a$' objectives based on false data, which could decrease the overall social welfare. Secondly, advertisers also need to take effort and time to reason out the optimal submission.

Denote $c_a$ by $a$'s true constraint and let $c_a'$ be her report. A delivery mechanism $\{(cpa_a^j,cv_a^j)\}$ is incentive-compatible (or strategy-proof) if expressing true constraint is a dominant strategy for $a$, no matter what the other advertisers do.
\begin{defn}
	A delivery mechanism $\{(cpa_a^j,cv_a^j)\}$ is incentive-compatible (IC) if $U_a(\{v_a^j\},\{(cpa_a^j,cv_a^j)\},c_a)\ge U_a(\{v_a^j\},\{(cpa_a^j,cv_a^j)\},c_a')$ for all $c_a$, all $c_a'$ and all $\{b_{-a}^j\}$.
\end{defn}

Since the advertising platform cannot force $a$ to adopt auto-bidding, we also require that $\{(cpa_a^j, cv_a^j)\}$ is individually rational, i.e., submitting true constraint to the advertising platform never suffers a loss.
\begin{defn}
	A delivery mechanism $\{(cpa_a^j,cv_a^j)\}$ is individually rational (IR) if $U_a(\{v_a^j\},\{(cpa_a^j,cv_a^j)\},c_a) \ge 0$ for all $c_a$ and all $\{b_{-a}^j\}$.
\end{defn}

An auto-bidding episode often lasts several days and the number of relevant ad requests matching advertiser $a$ is tens of thousands. Denote $cv_a(c_a)=\sum_{j\in R_a}cv_a^j(c_a)$ and $cpa_a(c_a)=\sum_{j}cpa_a^j(c_a)\cdot cv_a^j(c_a)/\sum_{j}cv_a^j(c_a)$ by the expected conversions and CPA realized in the auto-bidding episode, respectively. For practical reasons, we consider the following market model throughout this paper.
\begin{defn}[Intensive Auction Market]\label{intensive}
	An intensive auction market indicates that $cv_a$ and $cpa_a$ are differentiable and $cv_a(c_a)>0$ and $cpa_a(c_a)>0$ for any feasible $c_a$.
\end{defn}
In the rest of this paper, we design delivery mechanisms that satisfy IC and IR in intensive auction markets.

\section{Incentive Mechanisms for Advertisers with Target CPA Constraints}
According to different business models, ROI-constrained advertisers can have various expressions. In this section, we consider advertisers with tCPA expression. For this kind of advertisers, their value-per-conversion can be viewed as the same for all potential conversions. 

\subsection{The Double-Layer Framework}
For an advertiser $a$ employing auto-bidding, denote by $v_a$ as advertiser $a$'s value-per-conversion and $t_a\in (0,v_a)$ as her target CPA. In the beginning of the auto-bidding episode, $a$ is asked to submit $t_a$ to the advertising platform and based on the delivery mechanism $\{(cpa_a^j,cv_a^j)\}$, the platform implements the delivery $(cpa_a^j(t_a),cv_a^j(t_a))$ in each ad request $j$.
As the tCPA constraint is acting over the entire episode, it is not proper to deal with ad requests separately when designing IC and IR delivery mechanisms. To see this, we present an example to show that simply satisfying incentive compatibility in each single ad request does not suffice if the auto-bidding continues for a period.

\begin{example}[CPA-based second price auction]\label{example}
	Assume that there is only one available slot and once an advertiser wins the slot, a conversion occurs. A CPA-based second price auction allocates the slot to the advertiser with the highest tCPA report and charges her the second highest report.
\end{example}

It is easy to verify that when the advertising platform uses CPA-based second price auction, all advertisers will submit their true tCPA in each {\it single ad request}. Nonetheless, it is not the optimal strategy from the perspective of the whole episode. Consider a setting where $R_a=\{1,2\}$, $N_1=\{a,b,c,d\}$, $N_2=\{a,d\}$ and advertisers' tCPAs are $4, 5, 2, 1$, respectively. Suppose $a$'s value-per-conversion $v_a$ is higher than $5$, if she submits her true tCPA, she wins an impression only in ad request $2$ and achieves a profit of $v_a-1$ with an actual CPA $1<4$. However, if she submits a tCPA higher than $5$, she will win in both ad requests. In this case, her profit becomes $(v_a-5)+(v_a-1)$ with a CPA $\frac{5+1}{2}<4$. The above example indicates that designing IC delivery mechanisms should consider the correlation between different ad requests, which makes this problem very challenging. 

If advertiser $a$'s value-per-conversion is the same for all ad requests, her utility function (9) can be reformulated as:
\begin{alignat}{1}
(v_a-I_a\cdot cpa_a(t_a'))\cdot cv_a(t_a').
\end{alignat}
Based on this observation, we propose the following double-layer framework, which is illustrated in Figure \ref{double_layer}. The framework consists of two components. The first component (the green part) is in charge of designing the aggregated delivery mechanism $(cpa_a,cv_a)$.
The second component (the brown part) decomposes the aggregated requirement $(cv_a(t_a),cpa_a(t_a))$ into each ad request. Specifically, given a report $t_a'$ and an aggregated mechanism $(cpa_a,cv_a)$, the advertising platform splits $cpa_a(t_a')$ and $cv_a(t_a')$ into fine-grained $\{(cpa_a^j(t_a'),cv_a^j(t_a'))\}$, where $\sum_{j}cpa_a^j(t_a')\cdot cv_a^j(t_a')/\sum_{j}cv_a^j(t_a')=cpa_a(t_a')$ and $\sum_{j}cv_a^j(t_a')=cv_a(t_a')$. It is easy to verify that $\{(cpa_a^j,cv_a^j)\}$ is IC if and only if $(cv_a,cpa_a)$ is IC. 

Next, we characterize a necessary and sufficient condition for IC and IR $(cpa_a,cv_a)$, and analyze possible solutions in two scenarios. In the former scenario, advertiser $a$'s value-per-conversion is unknown to the advertising platform, while this kind of information is known in the latter. On top of that, we introduce a practical algorithm to decompose $(cpa_a,cv_a)$ into each ad request.

\begin{figure}[t]
	\centering
	\includegraphics[width=3in]{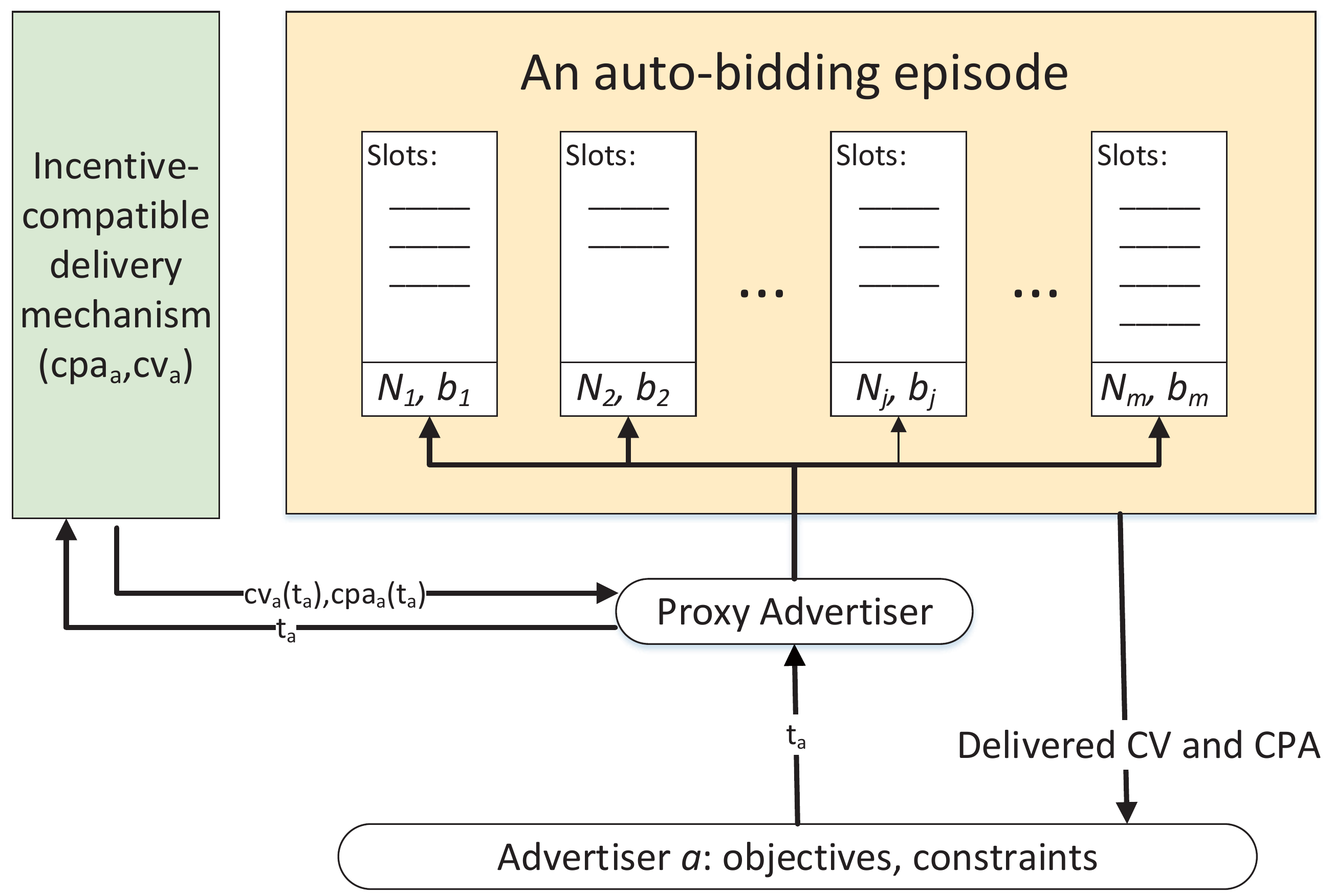}
	\caption{The double-layer framework for auto-bidding advertisers with tCPA expression.}\label{double_layer}
\end{figure}

\subsection{Incentive-Compatible Delivery Mechanisms}
At the beginning of the auto-bidding episode, the advertising platform only learns the advertiser's report $t_a'$. However, the term $v_a$ is also crucial for modeling advertiser $a$. We prove that if $v_a$ is unknown to the platform, only trivial mechanism exists for profit-maximizing advertisers. For the known case, the advertising platform can customize the delivery mechanism based on $v_a$ and other commercial considerations. Theorem \ref{main_result} identifies all $(cv_a,cpa_a)$ that is IC and IR for advertisers with tCPA expression, no matter whether $v_a$ is known or not.
\begin{theorem}\label{main_result}
	The cumulative delivery mechanism $(cpa_a,cv_a)$ is IC and IR if and only if $cpa_a(t_a)=t_a$ and $(v_a-I_a\cdot t_a)\cdot cv_a(t_a)$ is non-decreasing with $t_a$.
\end{theorem}
\begin{proof}
	(``$\Rightarrow$'') If $(cpa_a,cv_a)$ is IC and IR, then advertiser $a$'s utility $(v_a-I_a\cdot cpa_a(t_a))\cdot cv_a(t_a)$ must be non-decreasing in $t_a$. To see this, consider a setting where $t^l<t^h$ but $(v_a-I_a\cdot cpa_a(t^l))\cdot cv_a(t^l)>(v_a-I_a\cdot cpa_a(t^h))\cdot cv_a(t^h)$. Since $cpa_a(t^l)\le t_l< t^h$ (derived from the IR property), advertiser with a tCPA $t^h$ would under-report her constraint to achieve a higher utility. Due to this, we have that reporting $t_a$ as high as possible is always an optimal strategy as long as $cpa_a(t_a')\le t_a$, which further implies that $cpa_a(t_a)=t_a$ since $(cpa_a,cv_a)$ is incentive-compatible. In addition, replacing $cpa_a(t_a)=t_a$ in the utility function, we have that $(v_a-I_a\cdot t_a)\cdot cv_a(t_a)$ should be non-decreasing in $t_a$.
	
	(``$\Leftarrow$'') If $cpa_a(t_a)=t_a$, then $a$'s utility function becomes $(v_a-I_a\cdot t_a)\cdot cv_a(t_a)$. Since $(v_a-I_a\cdot t_a)\cdot cv_a(t_a)$ is non-decreasing with $t_a$ and the CPA delivered in the auto-bidding episode equals $a$'s report, it is best for $a$ to submit $t_a$. Individual rationality is satisfied since $t_a\in (0,v_a)$.
\end{proof}

The above result implies that if the advertising platform wants to acquire advertisers' true tCPA, the delivered CPA in the auto-bidding episode should equal their reports. This also explains why the delivery mechanism in Example \ref{example} is not IC. Given a delivery strategy $cpa_a(t_a)=t_a$, according to Theorem \ref{main_result}, if $cv_a(t_a)$ further guarantees that $a$'s utility is non-decreasing in her report, then $a$'s best strategy is submitting her true tCPA. The problem raised here is whether such a $cv_a$ exists. For a revenue-maximizing advertiser ($I_a=0$), her utility function is $v_a\cdot cv_a(t_a)$. Therefore, $(cpa_a(t_a)=t_a,cv(t_a))$ is IC and IR as long as $cv_a'>0$, no matter whether the advertising platform knows $v_a$ or not. If $a$'s goal is to maximize profit ($I_a=1$), we show that only trivial $cv_a$ works when the platform has no information about $v_a$.
\begin{prop}\label{unknown_value} For a profit-maximizing advertiser $a$, if $v_a$ is unknown to the advertising platform, $cv_a(t_a)=+\infty$ for any incentive-compatible $(cpa_a,cv_a)$.
\end{prop}
\begin{proof}
	Since $v_a$ is unknown to the advertising platform, the design of $cv_a$ does not depend on $v_a$. If $(cpa_a,cv_a)$ is IC and IR, then reporting $t_a$ is a dominant strategy for all advertiser with any $v_a$ and any $t_a\in (0,v_a)$. Given an incentive-compatible $(cpa_a,cv_a)$, $a$'s profit curve is $U_a^p(t_a)=(v_a-t_a)\cdot cv_a(t_a)$. When $t_a\rightarrow v_a$, the term $v_a-t_a\rightarrow 0$. Incentive compatibility requires that $U_a^p(t_a)$ is non-decreasing in $(0,t_a)$ which implies that $cv_a(t_a)$ should be unbounded when $t_a\rightarrow v_a$. Due to the fact that $cv_a(t_a)>0$ for all $t_a\in (0,v_a)$, we must have $cv_a(t_a)=+\infty$ when $t_a\rightarrow v_a$. Since $v_a\in \mathcal{R}_{+}$, we have $cv_a(t_a)=+\infty$ for all $t_a\in \mathcal{R}_{+}$.
\end{proof}

In practice, the assumption of unknown value-per-conversion is very restrictive. This is because the advertising platform can collect a mass of data when interacting with advertisers. Based on these data and advanced learning techniques, they can estimate the empirical distribution of advertisers' values to make further decisions. In addition, there are many cases, especially for E-commerce advertisers, where the value-per-conversion is public information. For these reasons, $v_a$ can be seen as a prior to the advertising platform. With a known $v_a$, the advertising platform can customize $a$'s conversion function $cv_a$ to achieve IC.
\begin{prop}\label{known_value}
	If $v_a$ is known and $cpa_a(t_a)=t_a$, then $cv_a$ is incentive-compatible if and only if $cv_a(t_a)=\frac{g(t_a)}{v_a-I_a\cdot t_a}$, where $g'\ge 0$ and $g(t_a)>0$.
\end{prop}
\begin{proof}
	According to Theorem \ref{main_result}, $a$'s utility function $U_a(t_a)$ is $(v_a-I_a\cdot t_a)\cdot cv_a(t_a)$ and is non-decreasing with $t_a$ provided that $(cpa_a,cv_a)$ is incentive-compatible. Because that $v_a$ is known, we have that $cv_a(t_a)=\frac{U_a(t_a)}{v_a-I_a\cdot t_a}$, where $U_a'\ge 0$ and $U_a(t_a)>0$ for all $t_a\in (0,v_a)$.
\end{proof}

In other words, if $v_a$ is known to the advertising platform, designing an incentive-compatible $(cpa_a,cv_a)$ is equivalent to designing $a$'s utility curve. Besides, given an incentive-compatible $(cpa_a,cv_a)$, the utility of the advertising platform in the auto-bidding episode is $t_a\cdot cv_a(t_a)=\frac{t_a\cdot g(t_a)}{v_a-I_a\cdot t_a}$, which is align with advertiser $a$'s utility $g(t_a)$. In reality, the advertising platform is able to craft $a$' utility curve $U_a$, or equivalently the conversion curve $cv_a$, to benefit the advertisers, the platforms and also the users.

\subsection{Implementation of The Double-Layer Framework}
So far, we have characterized the high-level requirements for advertisers revealing true tCPA constraints. In the following, we explore how to implement these requirements in each ad request.
According to (10), $a$'s utility in the episode is pinned by $cpa_a(t_a)$ and $cv_a(t_a)$, hence any bid sequence $\{b_a^j\}$ that achieves $(cpa_a(t_a),cv_a(t_a))$ is considered as an optimal bidding strategy w.l.o.g. Given an advertiser $a$ and her report $t_a'$, we next introduce two elementary strategies to reach the target $cpa_a(t_a')$ and $cv_a(t_a')$.

{\bf Exterior Controlling Strategy (ECS)}: Every time a new ad request $j\in R_a$ arrives, the advertising platform firstly customizes the bidding environment faced by $\tilde{a}$ through a PID controller, which includes the matched advertiser set $N_j$, the set of available slots $S_j$, etc. The PID controller works for stabilizing the delivered conversions and CPA at the level of $cpa_a(t_a')$ and $cv_a(t_a')$ during the episode, based on the available data at the serving time. In each customized bidding environment, the proxy advertiser $\tilde{a}$ optimizes the bid $b_a^j$ with the objective of maximizing utility.

{\bf Interior Controlling Strategy (ICS)}: Suppose the cumulative conversions and the delivered CPA in the first $j-1$ ad requests are $x$ and $y$, respectively. The proxy advertiser establishes a plan $\{(cpa_a^k(t_a'),cv_a^k(t_a'))\}_{k\ge j}$ for the future impressions so as to the total conversions and CPA reach the target $(cpa_a(t_a'),cv_a(t_a'))$ automatically. For example, $\tilde{a}$ can commit an egality strategy, where the remaining conversions $cv_a(t_a')-x$ and the remaining spending $cv_a(t_a')\cdot cpa_a(t_a')-xy$ are split into each coming ad request $k$ equally, i.e., $cv_a^k(t_a')=\frac{cv_a(t_a')-x}{|R_a|-j+1}$ and $cpa_a^k(t_a')=\frac{cpa_a(t_a')\cdot cv_a(t_a')-xy}{cv_a(t_a')-x}$.
Given a plan $\{(cv_a^k(t_a'),cpa_a^k(t_a'))\}_{k\ge j}$, $\tilde{a}$ will bid with the target $(cpa_a^j(t_a'), cv_a^j(t_a'))$ in ad request $j$.

The two naive methods have their own pros and cons. On the one hand, ECS eases the burden on the proxy advertiser, yet it places too much emphasis on the PID controller. If the auto-bidding lasts long enough, the controller works very well. But in the worst cases, the controller can alter the bidding environment drastically. In these circumstances, the benefits of other advertisers can be harmed and the overall social welfare would be reduced as well. On the other hand, ICS brutally incorporates the incentive-compatible conditions into each ad request without modifying the bidding environment. If everything goes well as planned, the target $(cpa_a(t_a'),cv_a(t_a'))$ will be reached automatically. Nonetheless, if the split strategy is not well crafted, $(cpa_a^j(t_a'),cv_a^j(t_a'))$ may be too strict to be satisfied.

Our solution depicted in Algorithm \ref{impl} is a combination of the above two methods. 
Specifically, in each ad request $j$ we have a target $(tcpa_a^j,tcv_a^j)$ (line 4), which is derived from an interior controlling strategy. Instead of satisfying $(tcpa_a^j,tcv_a^j)$ strictly, the proxy advertiser will optimize the utility function with constraints $|cpa_a^j(b_j)-tcpa_a^j|\le \alpha_j$ and $|cv_a^j(b_j)-tcv_a^j|\le \beta_j$, where $\alpha_j$ and $\beta_j$ are two parameters produced from the interior controlling strategy.
After observing the actual output (line 6), a PID controller will make minor adjustments on the bidding environment of next matched ad request (line 8), ensuring that the bidding process is in the right direction.

\begin{algorithm}[t]
	\begin{small}
		\SetKwInOut{Input}{\textbf{Input}}\SetKwInOut{Output}{\textbf{Output}}
		
		\Input{$\{(cpa_a,cv_a),I_a,t_a\}$} \Output{$\{(cpa_a^j,cv_a^j)\}$}
		\BlankLine
		initialize the delivery set $dy=\{\emptyset\}$\;
		\For{$j\leftarrow 1$ \KwTo $R_a$}{
			estimate others bids $\tilde{b}_{-a}^{j}$ based on $S_j,N_j$ and $j$\;
			compute a target $(tcpa_a^j,tcv_a^j)$ and errors $(\alpha_j,\beta_j)$ according to an interior controlling strategy $ICS(\tilde{b}_{-a}^{j},dy,cpa_a(t_a),cv_a(t_a))$\;
			set a bid $b_a^{j*}$ that solves the following \:
			$$max_{b_a^j}\quad (v_a^j-I_a\cdot cpa_a^j(b_a^j,\tilde{b}_{-a}^{j}))\cdot cv_a^j(b_a^j,\tilde{b}_{-a}^{j})$$
			$$s.t.\quad |cpa_a^j(b_a^j,\tilde{b}_{-a}^{j})-tcpa_a^j|\le \alpha_j$$
			$$\quad |cv_a^j(b_a^j,\tilde{b}_{-a}^{j})-tcv_a^j|\le \beta_j$$\\
			$(cpa_a^j,cv_a^j)=Auction(b_a^{j*},\{b_i^j\}_{i\in N_j})$\;
			$dy=dy\cup \{(cpa_a^j,cv_a^j)\}$\;
			adjust $(S_{j+1},N_{j+1})$ through an exterior controlling strategy 
			$ECS(dy,cpa_a(t_a),cv_a(t_a),j+1)$\;
		}
		\caption{{\small Auto-bidding Strategy for $(cpa_a,cv_a)$}}\label{impl}
	\end{small}
\end{algorithm}

\section{Incentive Mechanisms for Advertisers with Target ROI Constraints}
The previous section studies mechanisms for advertisers with tCPA expression, where advertisers' value-per-conversion is the same for all potential conversions. For this kind of advertisers, designing an incentive-compatible delivery mechanism $\{(cpa_a^j,cv_a^j)\}$ can be reduced to designing an incentive-compatible aggregated mechanism $(cpa_a,cv_a)$. In this section, we investigate a more complex scenario, in which advertisers' value-per-conversion varies with ad requests. For this setting, submitting a target CPA is not sufficient to express advertisers' ROI requirement, and the advertising platform needs to learn their target ROI to optimize the bidding process. 

Given an advertiser $a$ adopting auto-bidding, let $R_a^h\subseteq R_a$ be the ad request set with the same value-per-conversion $v_a^h$, i.e., $v_a^{j_1}=v_a^{j_2}=v_a^h$ for any $j_1,j_2\in R_a^h$. We have $R_a=\cup \{R_a^h\}_{v_a^h\in V_a}$ and $R_a^{h_1}\cap R_a^{h_2}=\emptyset$ for any $v_a^{h_1},v_a^{h_2} (h_1\ne h_2)\in V_a$.
Notice that in the objective function (9), advertiser $a$' value-per-conversion $v_a^h$ is entangled with the term $\sum_{j\in R_a^h} cv_a^j(\gamma_a)$, therefore $a$'s utility function cannot be simplified as clear as (10). In order to characterize strategy-proof delivery mechanisms for advertisers with tROI constraints, the advertising platform has to design each delivery function $(cpa_a^j,cv_a^j)$ carefully.
\begin{prop}\label{sub_result}
	The delivery mechanism $\{(cpa_a^j,cv_a^j)\}$ is IC and IR if and only if $\frac{\sum_{j\in R_a}(v_a^j-cpa_a^j(\gamma_a))\cdot cv_a^j(\gamma_a)}{\sum_{j\in R_a}cpa_a^j(\gamma_a)\cdot cv_a^j(\gamma_a)}=\gamma_a$ and the utility function is non-increasing with $\gamma_a$.
\end{prop}
\begin{proof}
	The proof is similar to that for Theorem \ref{main_result}. Firstly, advertiser $a$'s utility should be non-increasing with her report, otherwise she has incentives to submit a higher tROI. Therefore, in order to incentivize $a$ to act truthfully, the ROI delivered in the episode needs to equal her report. In addition, if the actual ROI achieved in auto-bidding is $a$'s report and her utility is also non-increasing in her report, then it is a dominant strategy for $a$ to submit her true tROI constraint.
\end{proof}

Although Proposition \ref{sub_result} gives a characterization of IC and IR delivery mechanisms, it is not trivial for the advertising platform to identify one practical solution, as each $(cpa_a^j,cv_a^j)$ should be designed together. Since the ad request set $R_a^h$ constitutes a short-term auto-bidding episode with the same value-per-conversion $v_a^h$, we next provide two families of delivery mechanisms that are easy to implement in reality. The first one is straightforward: divide the whole market $R_a$ into each sub-market $R_a^h$ and design incentive-compatible delivery mechanisms for each of them separately. Let $(cpa_{a^h},cv_{a^h})$ be the aggregated delivery mechanism in sub-market $R_a^h$, that is $cpa_{a^h}(\gamma_a)=\sum_{j\in R_a^h}cpa_a^j(\gamma_a)\cdot cv_a^j(\gamma_a)/\sum_{j\in R_a^h}cv_a^j(\gamma_a)$ and $cv_{a^h}(\gamma_a)=\sum_{j\in R_a^h}cv_a^j(\gamma_a)$.
\begin{corollary}\label{submarket_gamma}
	For each sub-market $R_a^h$, $(cpa_{a^h},cv_{a^h})$ is IC and IR if and only if $cpa_{a^h}(\gamma_a)=\frac{v_a^h}{1+\gamma_a}$ and $\frac{1+\gamma_a-I_a}{1+\gamma_a}\cdot cv_{a^h}(\gamma_a)$ is non-increasing with $\gamma_a$.
\end{corollary}
\begin{proof}
	For sub-market $R_a^h$, advertiser $a$'s value-per-conversion is $v_a^h$ for all conversions. Therefore, $a$'s utility function in $R_a^h$ can be expressed as $(v_a^h-I_a\cdot cpa_{a^h}(\gamma_a))\cdot cv_{a^h}(\gamma_a)$. According to Proposition \ref{sub_result}, we have $cpa_{a^h}(\gamma_a)=\frac{v_a^h}{1+\gamma_a}$ and now $a$'s utility function becomes $v_a^h\cdot\frac{ 1+\gamma_a-I_a}{1+\gamma_a}\cdot cv_a^h(\gamma_a)$. Since $a$'s utility function should be non-increasing with $\gamma_a$, we get the corollary.
\end{proof}

Since $(cpa_{a^h},cv_{a^h})$ is an aggregated delivery mechanism in sub-market $R_a^h$, it can be easily carried out by approaches developed for the tCPA setting. If $(cpa_{a^h},cv_{a^h})$ is IC and IR for sub-market $R_a^h$, we have that $\{(cpa_{a^h},cv_{a^h})\}$ is IC and IR for the auto-bidding episode.
\begin{theorem}\label{combin_solution}
	Given a delivery mechanism $\{(cpa_{a^h},cv_{a^h})\}$, if $(cpa_{a^h},cv_{a^h})$ is IC and IR in each sub-market $R_a^h\in R_a$, then $\{(cpa_{a^h},cv_{a^h})\}$ is IC and IR in $R_a$.
\end{theorem}
\begin{proof}
	If $(cpa_{a^h},cv_{a^h})$ is incentive-compatible in $R_a^h$, we have $cpa_{a^h}(\gamma_a)=\frac{v_a^h}{1+\gamma_a}$ and the actual ROI delivered in the episode equals exactly $\gamma_a$. Since $a$'s utility in each sub-market $R_a^h$ is non-increasing with $\gamma_a$, it is also true for the cumulative utility of the auto-bidding episode. According to Proposition \ref{sub_result}, the delivery mechanism $\{(cpa_{a^h},cv_{a^h})\}$ is incentive-compatible and individually rational.
\end{proof}

The first solution is simple and intuitive, and is applicable for performance advertisers with homogeneous products, such as online booksellers. However, each $R_a^h$ is treated as an independent market. In some scenarios, advertiser $a$ would like to launch a campaign that links each sub-market to balance the inventories of different products, for example her needs may look like ``I want more conversions in sub-markets with a higher value-per-conversion.'' or ``I want to maintain a certain percentage of conversions in two specific sub-markets''. For this kind of requirements, the interdependencies of delivery mechanisms in different sub-markets should be considered. 
In order to satisfy advertisers' needs of this kind, we propose another class of delivery mechanisms. Our solution starts with a decomposition of an aggregated conversion function $cv_a$, which is used to disentangle $v_a^j$ and $cv_a^j$. Then, we select out proper $cv_a$ to achieve incentive compatibility.
\begin{defn}[Linear Decomposition of $cv_a$]
	A linear decomposition of $cv_a$ is defined by a set of sub functions $\{k_h\cdot cv_a\}_{v_a^h\in V_a}$ where $k_h\ge 0$ and $\sum k_h=1$.
\end{defn}
A linear decomposition decouples the cumulative conversion function $cv_a(\gamma_a)$ into each sub-market. The parameter $k_h$ is preset according to advertiser $a$'s needs. 
Suppose the delivery mechanism in each sub-market $R_a^h$ is $(cpa_a^h,k_h\cdot cv_a)$, our next result identifies all IC delivery mechanisms.
\begin{theorem}\label{various_values}
	Given a cumulative conversion function $cv_a$ and a linear decomposition $\{k_h\cdot cv_a\}$, if $\sum_{h}v_a^h \cdot k_h=(1+\gamma_a)\cdot \sum_{h}cpa_a^h(\gamma_a)\cdot k_h$ and $\frac{1+\gamma_a-I_a}{1+\gamma_a}\cdot cv_a(\gamma_a)$ is non-increasing with $\gamma_a$, then the delivery mechanism $\{(cpa_a^h,k_h\cdot cv_a)\}$ is IC and IR.
\end{theorem}
\begin{proof}
	Given a delivery mechanism $\{(cpa_a^h,k_h\cdot cv_a)\}$, suppose $\sum_{h}v_a^h\cdot k_h=(1+\gamma_a)\cdot \sum_{h}cpa_a^h(\gamma_a)\cdot k_h$ and $\frac{1+\gamma_a-I_a}{1+\gamma_a}\cdot cv_a(\gamma_a)$ is non-increasing with $\gamma_a$. If $(cpa_a^h(\gamma_a),cv_{a^h}(\gamma_a))$ is implemented in each sub-market $R_a^h$, it is easy to verify that the actual ROI delivered in the auto-bidding episode is exactly $\gamma_a$ and $a$'s cumulative utility can be expressed as $\frac{1+\gamma_a-I_a}{1+\gamma_a}\cdot cv_a(\gamma_a) \cdot \sum_{h}v_a^h\cdot k_h$. Notice that $\sum_{h}v_a^h\cdot k_h$ is a positive constant, and hence advertiser $a$'s utility is non-increasing with her report. According to Proposition \ref{sub_result}, the delivery mechanism $\{(cpa_a^h,k_h\cdot cv_a)\}$ is incentive-compatible and individually rational.
\end{proof}
We emphasize that although $\{(cpa_{a^h},k_h\cdot cv_a)\}$ is IC and IR in the auto-bidding episode, it does not guarantee that each $(cpa_{a^h},k_h\cdot cv_a)$ is IC and IR in $R_a^h$. For instance, the delivery mechanism $\{(\frac{\sum_{h}v_a^h k_h}{1+\gamma_a},k_h\cdot cv_a)\}$ is IC for $R_a$ while $(\frac{\sum_{h}v_a^h k_h}{1+\gamma_a},k_h\cdot cv_a)$ is generally not IC for $R_a^h$. In addition, the two solutions proposed in Theorem \ref{combin_solution} and Theorem \ref{various_values} can work together for advertisers with more complex needs.

\section{Discussion and Conclusions}
\begin{figure}[t]
	\centering
	\includegraphics[width=3.3in]{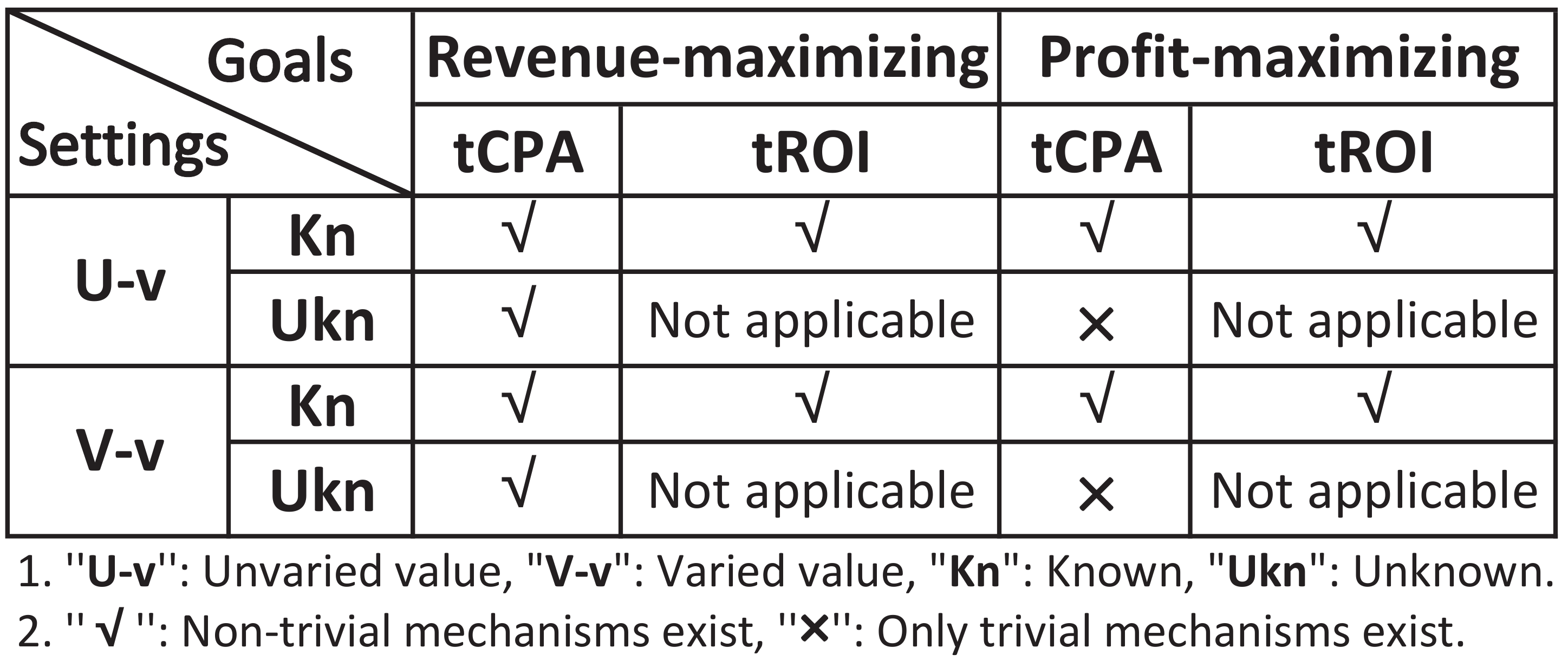}
	\caption{A panorama of IC and IR delivery mechanisms for ROI-constrained advertisers.}\label{results}
\end{figure}
Based on massive amount of data and advanced AI techniques, auto-bidding can provide better service for both the advertisers and the users. In order to optimize the auto-bidding process, advertisers are asked to provide their true demands to the platforms. But they have no incentives to follow the rule unless their own goals are met. To eliminate strategic behaviors of advertisers, the auto-bidding strategy and the incentive mechanism should be co-designed. This paper is the first research that considers this problem, which not only builds a general framework to study the incentive mechanisms, but also provides several practical solutions for various types of advertisers. One important future work is the study of dynamics of delivery mechanisms, including real-time strategy adjustment, convergence analysis, etc. Another interesting setting is that the utility function involves a cost term that is increasing in the constraint violation. It is also worth investigating how different delivery mechanisms affect the auto-bidding strategy.
\bibliographystyle{named}
\bibliography{manuscript}

\end{document}